\RequirePackage{rotating}
\documentclass[12pt]{article}

\usepackage{tikz}
\usetikzlibrary{arrows}

\usepackage[margin=1in]{geometry}

\usepackage{amsmath}
\usepackage{amssymb}
\usepackage{url}
\usepackage{MnSymbol}

\usepackage{rotating}
\usepackage{mathdots}

\usepackage[colorlinks,pdfborder={0 0 0},linkcolor={red!80!black},citecolor={green!50!black},urlcolor={blue!80!black}]{hyperref}
\usepackage{amsthm}
\usepackage[capitalize]{cleveref}

\usepackage[all,cmtip]{xy}

\newtheorem{theorem}{Theorem}[section]

\newtheorem{corollary}[theorem]{Corollary}

\newcounter{claims}[theorem]
\newtheorem{claim}[claims]{Claim}

\theoremstyle{definition}
\newtheorem{definition}[theorem]{Definition}

\newtheorem*{question*}{Question}

\newcounter{cases}[theorem]
\newtheorem{case}[cases]{Case}

\theoremstyle{remark}

\newtheorem*{remark*}{Remark}

\newcommand{\fs}{2^{< \omega}}
\newcommand{\cs}{2^{ \omega}}
\newcommand{\N}{\mathbb{N}}

\newcommand{\solr}{\preceq_S}

\DeclareMathOperator{\dom}{dom}

\renewcommand\labelenumi{(\arabic{enumi})}
\renewcommand\theenumi\labelenumi

\begin{document}

\title{Some Questions of Uniformity in Algorithmic Randomness}
\author{Laurent Bienvenu, Barbara F. Csima, Matthew Harrison-Trainor}

\maketitle

\abstract{The $\Omega$ numbers---the halting probabilities of universal prefix-free machines---are known to be exactly the Martin-L\"of random left-c.e.\ reals. We show that one cannot uniformly produce, from a Martin-L\"of random left-c.e.\ real $\alpha$, a universal prefix-free machine $U$ whose halting probability is $\alpha$. We also answer a question of Barmpalias and Lewis-Pye by showing that given a left-c.e.\ real $\alpha$, one cannot uniformly produce a left-c.e.\ real~$\beta$ such that $\alpha - \beta$ is neither left-c.e.\ nor right-c.e.}

\section{Introduction}

Prefix-free Kolmogorov complexity, which is perhaps the most prominent version of Kolmogorov complexity in the study of algorithmic randomness, is defined via prefix-free machines: A  prefix-free machine is a partial computable function $M: \fs \rightarrow \fs$ ($\fs$ being the set of finite binary strings) such that no two distinct elements of $\dom(M)$ are comparable under the prefix relation. The prefix-free Kolmogorov complexity of $x \in \fs$ relative to the machine~$M$ is defined to be the quantity $K_M(x)=\min \{|p| \, : \, M(p)=x\}$. To get a machine-independent notion of Kolmogorov complexity, one needs to take an optimal prefix-free machine, that is, a prefix-free machine $U$ such that for any prefix-free machine $M$, one has $K_U \leq K_M + c_M$ for some constant~$c_M$ which depends solely on~$M$. Then one defines the prefix-free Kolmogorov complexity $K$ by setting $K=K_U$. The resulting function~$K$ only depends on the choice of~$U$ by an additive constant, because by definition, if $U$ and $V$ are optimal machines, then $|K_U-K_V|=O(1)$. To be complete, one needs to make sure optimal machines exist. One way to build one is to take a total computable function $e \mapsto \sigma_e$ from $\N$ to $\fs$ whose range is prefix-free (for example, $\sigma_e=0^e1$), and set $U(\sigma_e \tau)=M_e(\tau)$ where $(M_e)$ is an effective enumeration of all prefix-free machines. It is easy to see that $U$ is prefix-free and for all~$e$, $K_U \leq K_{M_e} + |\sigma_e|$, hence $U$ is optimal. Machines~$U$ of this type are called \emph{universal by adjunction} and they form a strict subclass of optimal prefix-free machines.\footnote{For example, given a universal prefix-free machine $U$, we can construct an optimal prefix-free machine $V$, which is not universal by adjunction, by defining, for $p \in \dom(U)$, $V(p0) = V(p1) = U(p)$ if $|p|$ odd, and $V(p) = U(p)$ if $|p|$ is even. This is well-defined because $U$ is prefix-free, and the fact that $U$ is prefix-free and optimal implies that $V$ is. $V$ is not universal by adjunction; one can see this for example by noting that every string in the domain of $V$ is of even length, but this is not true of any machine that is universal by adjunction. See, for example, \cite{CNSS,CaludeStaiger}.}

\begin{remark*}
	Often no distinction is made between optimal prefix-free machines and universal prefix-free machines. E.g., in \cite{Nies2009} it is said that optimal prefix-free machines are often called universal prefix-free machines. In this paper, the distinction will be important. An optimal prefix-free machine is a prefix-free machine~$U$ such that for every prefix-free machine $M$, there is a constant $c_M$ such that $K_U \leq K_M + c_M$. A universal prefix-free machine is one that is universal by adjunction. Thus every universal machine is optimal, but the converse is not true. Every machine in this paper will be prefix-free, and so we often omit the term `prefix-free'.
\end{remark*}

\subsection{Omega Numbers}

Given a prefix-free machine~$M$, one can consider the `halting probability' of~$M$, defined by
\[
\Omega_M=\sum_{M(\sigma)\downarrow} 2^{-|\sigma|}.
\]
The term `halting probability' is justified by the following observation: a prefix-free machine $M$ can be naturally extended to a partial functional from $\cs$, the set of infinite binary sequences, to $\fs$, where for $X \in \cs$, $M(X)$ is defined to be $M(\sigma)$ if some $\sigma \in \dom(M)$ is a prefix of $X$, and $M(X) \uparrow$ otherwise. The prefix-freeness of $M$ on finite strings ensures that this extension is well-defined. With this point of view, $\Omega_M$ is simply $\mu\{ X \in \cs : M(X) \downarrow\}$, where $\mu$ is the uniform probability measure (a.k.a.\ Lebesgue measure) on $\cs$, that is, the measure where each bit of $X$ is equal to $0$ with probability $1/2$ independently of all other bits.

For any machine $M$, the number $\Omega_M$ is fairly simple from a computability-theoretic viewpoint, namely, it is the limit of a computable non-decreasing sequence of rationals (this is easy to see, because $\Omega_M$ is the limit of $\Omega_{M_s}=\sum_{M(\sigma)[s]\downarrow} 2^{-|\sigma|}$).  We call such a real \emph{left-c.e.} It turns out that every left-c.e.\ real $\alpha \in [0,1]$ can be represented in this way, i.e., for any left-c.e.~$\alpha \in [0,1]$, there exists a prefix-free machine~$M$ such that $\alpha = \Omega_M$, as consequence of the Kraft-Chaitin theorem (see~\cite[Theorem 3.6.1]{DowneyH2010}).

One of the first major results in algorithmic randomness was Chaitin's theorem~ \cite{Chaitin75} that the halting probability $\Omega_U$ of an optimal machine~$U$ is always an algorithmically random real, in the sense of Martin-L\"of (for background on Martin-L\"of randomness, one can consult~\cite{DowneyH2010,Nies2009}). From here on we simply call a real \emph{random} if it is random in the sense of Martin-L\"of.

This is particularly interesting because this gives ``concrete'' examples of Martin-L\"of random reals, which furthermore are, as we just saw, left-c.e. Whether the converse is true, that is, whether every random left-c.e.\ real $\alpha \in [0,1]$ is equal to $\Omega_U$ for \emph{some} optimal machine~$U$ remained open for a long time. The answer turned out to be positive, a remarkable result with no less remarkable history. Shortly after the work of Chaitin, Solovay~\cite{Solovay75} introduced a preorder on left-c.e.\ reals, which we now call Solovay reducibility: for $\alpha, \beta$ left-c.e., we say that $\alpha$ is Solovay-reducible to $\beta$, which we write $\alpha \solr \beta$, if for some positive integer~$n$, $n \beta - \alpha$ is left-c.e.\footnote{In fact Solovay gave a more intuitive definition, which in substance states that computable approximations of $\beta$ from below converge more slowly than computable approximations of $\alpha$ from below, but the version we give is equivalent to Solovay's original definition and easier to manipulate.}. Solovay showed that reals of type $\Omega_U$ for optimal~$U$ are maximal with respect to the Solovay reducibility. While this did not fully settle the above question, Solovay reducibility turned out to be the pivotal notion towards its solution. Together with Solovay's result, subsequent work lead to the following theorem.

\begin{theorem}\label{thm:equiv-omega}
For $\alpha \in [0,1]$ left-c.e., the following are equivalent.
\begin{itemize}
\item[(a)] $\alpha$ is Martin-L\"of random
\item[(b)] $\alpha = \Omega_U$ for some optimal machine~$U$
\item[(c)] $\alpha$ is maximal w.r.t\ Solovay reducibility.
\end{itemize}
\end{theorem}

The implication $(b) \Rightarrow (a)$ is Chaitin's result and the implication $(b) \Rightarrow (c)$ is Solovay's, as discussed above. Calude, Hertling, Khoussainov, and Wang~\cite{CaludeHertlingKhoussainovWang} showed $(c) \Rightarrow (b)$, and the last crucial step $(a) \Rightarrow (c)$ was made by Ku\v cera and Slaman~\cite{KuceraS2001}. We refer the reader to the survey~\cite{BienvenuS2012} for an exposition of this result.\\

Summing up what we know so far, we have for any real $\alpha \in [0,1]$:
\begin{eqnarray*}
\alpha~\text{is left-c.e.} & \Leftrightarrow & \alpha=\Omega_M~\text{for some machine}~M\\
\alpha~\text{is left-c.e.\ and random} & \Leftrightarrow & \alpha=\Omega_U~\text{for some optimal machine}~U
\end{eqnarray*}

The first equivalence is uniform: Given a prefix-free machine $M$ (represented by its index in an effective enumeration of all prefix-free machines), we can pass in a uniform way to a left-c.e.\ index for $\Omega_M$; and moreover, given a left-c.e.\ index for a left-c.e.\ real $\alpha \in [0,1]$, we can pass uniformly to an index for a prefix-free machine $M$ with $\Omega_M = \alpha$ (a consequence of the so-called Kraft-Chaitin theorem, see~\cite[Theorem 3.6.1]{DowneyH2010}). By a left-c.e.\ index, we mean an index for a non-decreasing sequence of rationals

It was previously open however (see for example~\cite[p.11]{Barmpalias}) whether the second equivalence was uniform, that is: given an index for a random left-c.e.\ $\alpha \in [0,1]$, can we uniformly obtain an index for an \emph{optimal} machine~$U$ such that $\alpha=\Omega_U$? Our first main result is a negative answer to this question.

\begin{theorem}\label{no-uniform-machine-constr}
	There is no partial computable function $f$ such that if $e$ is an index for a Martin-L\"of random left-c.e.\ real $\alpha \in [0,1]$, then  the value of $f(e)$ is defined and is an index for an optimal Turing machine $M_{f(e)}$ with halting probability $\alpha$.
\end{theorem}

\noindent Thus one cannot uniformly view a Martin-L\"of random left-c.e.\ real as an $\Omega$ number. \\


On the other hand, we show that given a left-c.e.\ random $\alpha \in [0,1]$, one can uniformly find a universal left-c.e.\ semi-measure $m$ with $\sum_i m(i) = \alpha$. An interesting corollary is that one cannot uniformly turn a universal left-c.e.\ semi-measure~$m$ into a universal machine whose halting probability is $\sum_i m(i)$.

\subsection{Differences of left-c.e.\ reals}

The set of left-c.e.\ reals is closed under addition and multiplication, not under subtraction or inverse. However, the set $\{\alpha - \beta \mid \alpha, \beta\ \text{left-c.e.}\}$, of differences of two left-c.e.\ reals is algebraically much better behaved, namely it is a real closed field~\cite{AWZ,Raichev,ng}. Barmpalias and Lewis-Pye proved the following theorem.

\begin{theorem}[Barmpalias and Lewis-Pye \cite{BarmpaliasLewis}]\label{thm:bl1}
	If $\alpha$ is a non-computable left-c.e.\ real there exists a left-c.e.\ real $\beta$ such that $\alpha - \beta$ is neither left-c.e.\ nor right-c.e.
\end{theorem}

\noindent The proof is non-uniform, and considers two separate cases depending on whether or not $\alpha$ is Martin-L\"of random (though it is uniform in each of these cases). Barmpalias and Lewis-Pye ask whether there is a uniform construction; we show that the answer is negative.

\begin{theorem}\label{thm:no-unif-diff}
	There is no partial computable function $f$ such that if $e$ is an index for a non-computable left-c.e.\ real $\alpha$, then $f(e)$ is defined and is an index for a left-c.e.\ real $\beta$ such that $\alpha - \beta$ is neither left-c.e.\ nor right-c.e.
\end{theorem}

\noindent Barmpalias and Lewis-Pye note that it follows from \cite[Theorem 3.5]{DowneyHirschfeldtNies} that if $\alpha$ and $\beta$ are left-c.e.\ reals and $\alpha$ is Martin-L\"of random while $\beta$ is not, then $\alpha - \beta$ is a Martin-L\"of random left-c.e.\ real. In particular, if $\alpha$ in Theorem~\ref{thm:bl1} is Martin-L\"of random, then the corresponding $\beta$ must be Martin-L\"of random as well. Thus $\alpha$ and $\beta$ are the halting probabilities of universal machines.

\begin{theorem}[Barmpalias and Lewis-Pye \cite{BarmpaliasLewis}]\label{thm:bl2}
	For every universal machine~$U$, there is a universal machine~$V$ such that $\Omega_U-\Omega_V$ is neither left-c.e.\ nor right-c.e.
\end{theorem}

\noindent Recall that the construction for Theorem \ref{thm:bl1} was uniform in the Martin-L\"of random case. So it is not too surprising that Theorem \ref{thm:bl2} is uniform; but because we cannot pass uniformly from an arbitrary Martin-L\"of random left-c.e.\ real to a universal machine (Theorem \ref{no-uniform-machine-constr}), this requires a new proof.

\begin{theorem}\label{thm:unif-diff-mach}
	Theorem~\ref{thm:bl2} is uniform in the sense that there is a total computable function~$f$ such that if $U=M_e$ is an optimal (respectively universal by adjunction) machine, then $V=M_{f(e)}$ is optimal (respectively universal by adjunction) and $\Omega_{U}-\Omega_{V}$ is neither left-c.e.\ nor right-c.e.
\end{theorem}

\section{Omega Numbers}

\subsection{No uniform construction of universal machines}

We prove Theorem \ref{no-uniform-machine-constr}:

{
\renewcommand{\thetheorem}{\ref{no-uniform-machine-constr}}
\begin{theorem}
	There is no partial computable function $f$ such that if $e$ is an index for a random left-c.e.\ real $\alpha \in [0,1]$, then $f(e)$ is defined and is an index for an optimal prefix-free machine $M_{f(e)}$ with halting probability $\alpha$.
\end{theorem}
}
\addtocounter{theorem}{-1}
\begin{proof}
	First note that we can assume that the partial computable function $f$ is total. Indeed, define a total function $g$ as follows: for each input $e$, $g(e)$ is an index for a machine which on input~$\sigma$ waits for $f(e)$ to converge, and then copies $M_{f(e)}(\sigma)$.
	
	Fix a partial computable function $f$ taking indices for left-c.e.\ reals to indices for prefix-free machines. Using the recursion theorem, we will define a left-c.e.\ ML-random $\alpha = \alpha_e \in [0,1]$ using, in its definition, the index $f(e)$ of a prefix-free Turing machine $M_{f(e)}$. We must define $\alpha$ even if $M_{f(e)}$ is not optimal or $f(e)$ does not converge. We can always assume that $M_{f(e)}$ is prefix-free by not letting it converge on a string $\sigma$ if it has already converged on a prefix of $\sigma$; we can also assume that $f(e)$ converges by having $\alpha$ follow some fixed left-c.e.\ random $\beta$ (say the one chosen below) until $f(e)$ converges. During the construction of $\alpha$ we will also build an auxiliary machine $Q$. We will ensure that $\alpha$ is a random left-c.e. real, but that either $M_{f(e)}$ is not optimal (which will happen because for all $d$, there is $\sigma$ such that $K_{M_{f(e)}}(\sigma) > K_Q(\sigma) + d$), or $\mu(\dom(M_{f(e)}))$ is not $\alpha$. This will prove the theorem.
	
	In the construction, we will build $\alpha = \alpha_e$ (using the recursion theorem to know the index~$e$ in advance) while watching $M = M_{f(e)}$. (From now on, we drop the index $e$ everywhere; we will write $\alpha_s$ for the left-c.e.\ approximation to $\alpha$.) We will try to meet the requirements:
	\[ \text{$R_{d}$ : For some $\sigma$, $K_{M}(\sigma) > K_Q(\sigma) + d$.}\]
	If $M$ is universal, then there must be some $d$ such that, for all $\sigma$, $K_M(\sigma) \leq K_Q(\sigma) + d$. Thus meeting $R_d$ for every $d$ will ensure that $M$ is not universal.
	At the same time, we will be trying to get a global win by having $\mu(\dom(M)) \neq \alpha$.
	
	We will define stage-by-stage rationals $\alpha_0 < \alpha_1 < \alpha_2 < \cdots$ with $\alpha = \lim_s \alpha_s$. (Recall that an index for such a sequence is an index for $\alpha$.) Fix $\beta$ a left-c.e.\ random, $\frac{3}{4} < \beta < 1$. We will have $\alpha = q\beta+l$ for some $q, l \in \mathbb{Q}$, $q > 0$, so that $\alpha$ will be random (indeed, multiplying by the denominator of~$q$ and subtracting $\beta$, we see that $\beta \solr \alpha$, and since $\beta$ is random, by Theorem~\ref{thm:equiv-omega}, so is $\alpha$). It is quite possible that we will have $\alpha = \beta$. Let $\beta_0 < \beta_1 < \beta_2 < \cdots$ be a computable sequence of rationals with limit $\beta$. At each stage $s$ we will define $\alpha_s = q_s\beta_s + l_s$ for some $q_s, l_s \in \mathbb{Q}$ in such a way that $q = \lim_s q_s$ and $l = \lim_s l_s$ are reached after finitely many stages. We think of our opponent as defining the machine $M$ with measure $\gamma_s$ at stage $s$, with $\gamma = \lim_s \gamma_s$ the measure of the domain of $M$. Our opponent must keep $\gamma_s \leq \alpha_s$, as if they ever have $\gamma_s > \alpha_s$ then we can immediately abandon the construction and choose $q,l$ such that $\alpha = q \beta + l$ has $\alpha_s < \alpha < \gamma$ and get a global win. Our opponent also has to (eventually) increase $\gamma_s$ whenever we increase $\alpha_s$, or they will have $\gamma < \alpha$. However, they may wait to do this. But, intuitively speaking, whenever we increase $\alpha_s$, we can wait for our opponent to increase $\gamma_s$ correspondingly (as long as, in the meantime, we work towards making $\alpha$ random).
	
	The requirements can be in one of four states: \textbf{inactive}, \textbf{preparing}, \textbf{waiting}, and \textbf{restraining}. Unless it is injured by a higher priority requirement, in which case it becomes \textbf{inactive}, a requirement will begin \textbf{inactive}, then be \textbf{preparing}, before switching back and forth between \textbf{waiting} and \textbf{restraining}.
	
	Before giving the formal construction, we will give an overview. To start, each requirement will be \textbf{inactive}. When activated, a requirement will be in state \textbf{preparing}. When entering state \textbf{preparing}, a requirement $R_d$ will have a \textit{reserved code} $\tau \in 2^{< \omega}$ and a \textit{restraint} $r_d = 2^{-(|\tau| + d)}$. The reserved code $\tau$ will be such that $Q$ has not yet converged on input $\tau$ nor on any prefix or extension of~$\tau$, so that we can still use $\tau$ as a code for some string $\sigma$ to make $K_Q(\sigma) \leq |\tau|$. While in this state, our left-c.e.\ approximation to $\alpha$ will copy that of $\beta$. The requirement $R_d$ will remain in this state until the measure of the domain of the machine $M$ is close to our current approximation to $\alpha$, namely, within $r_d$. (If our opponent does not increase the measure of $M$ as we increase the approximation to $\alpha$, then we win.) At this point, we will set $Q(\tau) = \sigma$ for some string $\sigma$ for which $K_M(\sigma)$ is currently greater than $|\tau|+d$. The requirement will move into state \textbf{waiting}. From now on, we are trying to ensure that $M$ can never converge on a string of length $\leq |\tau| + d$, so that $K_M(\sigma)$ will never drop below $|\tau|+d$, satisfying $R_d$. We do this by having the approximation to $\alpha_s$ grow very slowly, so that $M$ can only add a small amount of measure at each stage. $R_d$ will now move between the states \textbf{waiting} and \textbf{restraining}. The requirement $R_d$ will remain in state \textbf{waiting} at stages $s$ when the measure of the domain of $M$ is close (within $r_d$) to $\beta_s$, so that $R_d$ is content to have $\alpha$ approximate $\beta$. However, at some stages $s$, it might be that $\beta_s$ is at least $r_d$ greater than $\gamma_s$, the measure of the domain of $M$ so far. In this case, $R_d$ is in state \textbf{restraining} and has to actively restrain $\alpha_s$ to not increase too much. Letting $l=\alpha_{s-1}$ and $q = r_d - (\alpha_{s-1} - \gamma_s)$, where $s$ is the stage when $R_d$ enters the state \textbf{restraining}, $R_d$ has $\alpha$ temporarily approximate $q \beta + l$. Whenever the measure of the domain of $M$ increases by $\frac{1}{2} r_d$, $R_d$ updates the values of $q$ and $l$ (recall that $\beta \geq \frac{3}{4}$). Thus, each time the values of $q$ and $l$ are reset, the measure of the domain of $M$ has increased by at least $\frac{1}{2}r_d$. (Again, if our opponent does not increase the measure of $M$ as we increase the approximation to $\alpha$, then we win.) This can happen at most finitely many times until the measure of the domain of $M$ is within $r_d$ of the current approximation to $\beta$, and so the requirement re-enters state \textbf{waiting}.\footnote{Of course, the requirement does not have to re-enter state \textbf{waiting}, but in this case the values of $q$ and $l$ are eventually fixed.} The requirement may then later re-enter state \textbf{restraining} if the approximation to $\beta_s$ increases too much faster than the measure of the domain of $M$, but since the measure of the domain of $M$ will increase by at least $\frac{1}{2} r_d$ every time $R_d$ switches from \textbf{restraining} to \textbf{waiting}, $R_d$ can only switch finitely many times.
	
	Just considering one requirement, the possible outcomes of the construction are as follows:
	\begin{itemize}
		\item $\gamma_s > \alpha_s$ at some stage $s$, in which case we can immediately ensure that $\alpha < \gamma$ and that $\alpha$ is random.
		\item $\gamma < \alpha$; the requirement may get stuck in \textbf{preparing} or \textbf{restraining}. If it gets stuck in \textbf{preparing}, we have $\alpha = \beta$ is random. If it gets stuck in \textbf{restraining}, we have $\alpha = q \beta + l$, with $q$ and $l$ rational, and this is random.
		\item $\gamma = \alpha$; in this case, the requirement always leaves \textbf{preparing}, and every time it enters \textbf{restraining} it returns to \textbf{waiting}. After some stage, it is always in \textbf{waiting} and has $\alpha = \beta$, which is random. The requirement is satisfied by having $K_Q(\sigma) \leq |\tau|$ but $K_M(\sigma) > |\tau| + d$.
	\end{itemize}
	With multiple requirements, there is injury. A requirement only allows lower priority requirements to be active while it is \textbf{waiting}. Every stage at which a requirement is \textbf{preparing} or \textbf{restraining}, it injures all lower priority requirements. So, at any stage, there is at most one requirement---the lowest priority active requirement---which can be in a state other than \textbf{waiting}.
	
	\medskip{}
	
	\noindent \textit{Construction.}
	
	\medskip{}
	
	\noindent \textit{Stage $0$.} Begin with $\alpha_0 = 0$, all the requirements other than $R_0$ inactive, and $Q_0$ not converged on any input. 
	
	Set $\alpha_{s} = \beta_{s}$. Activate $R_{0}$ and put it in state \textbf{preparing}. Choose a reserved code $\tau_{0}$ such that $Q_s(\tau_{0}) \uparrow$ and set the restraint $r_{1} = 2^{-|\tau_{0}|}$.
	
	\medskip{}
	
	\noindent \textit{Stage $s > 0$.} Let $\gamma_s = \mu(\dom(M_s))$ be the measure of the domain of $M$ at stage $s$. If $\gamma_s > \alpha_{s-1}$, we can immediately end the construction, letting $\alpha_{t} = \alpha_{s-1} + (\gamma_s - \alpha_{s-1}) \beta_t$ for $t \geq s$, so that
	\[ \alpha = \lim_{t \to \infty} \alpha_t = \alpha_{s-1} + (\gamma_s - \alpha_{s-1}) \beta < \gamma_s \leq \mu(\dom(M_s)).\]
	So for the rest of this stage, we may assume that $\gamma_s \leq \alpha_{s-1}$.
	
	Find the highest-priority active requirement $R_d$, if it exists, such that $\beta_s - \gamma_s \geq r_d$. Cancel every lower priority requirement. Let $R_d$ be the lowest priority active requirement. (Every higher priority requirement is in state \textbf{waiting}.)
	
	\begin{case}
		$R_d$ is \textbf{preparing}.
	\end{case}
	
	\noindent Set $\alpha_{s} = \beta_s$. $R_d$ has a reserved code $\tau_d$ and restraint $r_d$. If $\beta_s - \gamma_s > r_d$, $R_d$ remains \textbf{preparing}. Otherwise, if $\beta_s - \gamma_s < r_d$, find a string $\sigma_d$ such that $K_M(\sigma_d)[s] > |\tau_d| + d$. Put $Q(\tau_d) = \sigma_d$. $R_d$ is now \textbf{waiting}.
	
	\begin{case}
		$R_d$ is \textbf{waiting} and $\beta_s - \gamma_s < r_d$.
	\end{case}
	
	\noindent Set $\alpha_{s} = \beta_{s}$. Requirement $R_d$ remains in state \textbf{waiting}. Activate $R_{d+1}$ and put it in state \textbf{preparing}. Choose a reserved code $\tau_{d+1}$ such that $Q_s(\tau_{d+1}) \uparrow$ and set the restraint $r_{d+1} = 2^{-|\tau_{d+1}| - d - 1}$.
	
	\begin{case}
		$R_d$ is \textbf{waiting} and $\beta_s - \gamma_s \geq r_d$.
	\end{case}
	
	\noindent Set the reference values $l_d = \alpha_{s-1}$ and $q_{d}= r_d - (\alpha_{s-1} - \gamma_s)$. (In Claim \ref{claim:alpha-well-defined} we will show that $q_d > 0$.) Put $R_d$ in state \textbf{restraining}. Set $\alpha_{s} = q_d \beta_s + l_d$.
	
	\begin{case}
		$R_d$ is in state \textbf{restraining}.
	\end{case}
	
	\noindent $R_d$ has a restraint $r_d$ and reference values $q_d$ and $l_d$. If $\gamma_s \leq l_d + \frac{1}{2} q_d$, keep the same reference values, and set $\alpha_{s} = q_d \beta_s + l_d$. If $\gamma_s > l_d + \frac{1}{2} q_d$, then what we do depends on whether $\beta_s - \gamma_s < r_d$ or $\beta_s - \gamma_s \geq r_d$. In either case, we call stage $s$ \textit{incremental for $R_d$}. If $\beta_s - \gamma_s < r_d$, then set $\alpha_s = \beta_s$ and put $R_d$ into state \textbf{waiting}. If $\beta_s - \gamma_s \geq r_d$, change the reference values $l_d$ and $q_d$ to $l_d = \alpha_{s-1}$ and $q_{d}= r_d - (\alpha_{s-1} - \gamma_s)$, and set $\alpha_{s} = q_d \beta_s + l_d$. $R_d$ remains \textbf{restraining}.
	
	\medskip{}
	
	\noindent \textit{End construction.}
	
	\medskip{}
	
	\noindent \textit{Verification.}
	\begin{claim}\label{claim:alpha-well-defined} At every stage $s>0$, $\alpha_{s-1}\leq \alpha_s \leq \beta_s$, and for every requirement $R_d$ which is active at stage $s$, either $R_d$ is \textup{\textbf{preparing}} or $\alpha_s - \gamma_s < r_d$.
\end{claim}
\begin{proof} Assume the result holds for all $t<s$. Let $d$ be the lowest priority active requirement at stage $s$ (after the cancellation). By choice of $d$, for $d' < d$ we have $\beta_s - \gamma_s < r_{d'}$. We now check that no matter which case of the construction was used to define $\alpha_s$, the result holds. In all cases we will have $\alpha_s - \gamma_s \leq \beta_s - \gamma_s < r_{d'}$, so it is really $\alpha_s - \gamma_s < r_d$ that we must check.
\begin{enumerate}
\item At stage $s$ the construction was in Case 1 or Case 2. We set $\alpha_s = \beta_s \geq \beta_{s-1} \geq \alpha_{s-1}$. Either we are in Case 1 and $R_d$ remains \textbf{preparing}, or $\alpha_s - \gamma_s= \beta_s - \gamma_s < r_{d}$.
\item At stage $s$ the construction was in Case 3. We set $\alpha_s = q_d\beta_s + l_d$. Now in Case 3, $l_d = \alpha_{s-1}$ and $q_d = (r_d - (\alpha_{s-1} - \gamma_s))$.  Note that $\alpha_{s-1} - \gamma_s \leq \alpha_{s-1} - \gamma_{s-1} < r_d$ by induction, so $\alpha_s \geq \alpha_{s-1}$. Also
    \begin{align*}\alpha_s &= (r_d - (\alpha_{s-1} - \gamma_s))\beta_s +\alpha_{s-1} \\
    & = r_d\beta_s - (\alpha_{s-1} - \gamma_s)\beta_s + (\alpha_{s-1} - \gamma_s) + \gamma_s \\
    &= r_d\beta_s +(1-\beta_s)(\alpha_{s-1} - \gamma_s) +\gamma_s \\
    & < r_d\beta_s +(1-\beta_s)r_d +\gamma_s \\
    & = r_d + \gamma_s \\
    & \leq \beta_s.\end{align*} Finally, since we've just seen that $\alpha_s < r_d + \gamma_s$, we have that $\alpha_s - \gamma_s < r_d$.
\item At stage $s$ the construction was in Case 4. We set $\alpha_s = q_d\beta_s + l_d$. Then since $R_d$ was in state \textbf{restraining} at stage $s$, we must have defined $\alpha_{s-1}= q_d\beta_{s-1} + l_d$ \emph{unless} $s$ was an incremental stage, in which case $q_d$ and $l_d$ were reset at stage $s$ before defining $\alpha_s$. If $s$ \emph{was not} incremental, then $\alpha_s = q_d(\beta_s - \beta_{s-1}) + q_d\beta_{s-1} + l_d = q_d(\beta_s - \beta_{s-1})+\alpha_{s-1} \leq q_d(\beta_s - \beta_{s-1})+\beta_{s-1} \leq \beta_s$. Also $\alpha_{s-1}= q_d\beta_{s-1} + l_d \leq q_d\beta_{s} + l_d = \alpha_s$. Finally, if we let $\tilde{s}<s$ be the stage where $q_d$ and $l_d$ were last defined, then we see that \begin{align*}\alpha_s - \gamma_s & = q_d\beta_s + l_d - \gamma_s \\
    & = (r_d - (\alpha_{\tilde{s}-1} - \gamma_{\tilde{s}}))\beta_s + \alpha_{\tilde{s}-1} - \gamma_{s} \\
    & \leq (r_d - (\alpha_{\tilde{s}-1} - \gamma_{\tilde{s}}))\beta_s + \alpha_{\tilde{s}-1} - \gamma_{\tilde{s}} \\
    &= r_d\beta_s + (1-\beta_s)(\alpha_{\tilde{s}-1} - \gamma_{\tilde{s}}) \\
    & < r_d.
    \end{align*}
    Now suppose stage $s$ was incremental for $R_d$. If $\beta_s - \gamma_s < r_d$, then the result follows as in (1), and if $\beta_s - \gamma_s \geq r_d$, then the result follows as in (3).\qedhere
\end{enumerate}
\end{proof}
	\begin{claim}\label{claim:fin-exp}
		Suppose that the requirement $R_d$ is activated at stage $s$ and never injured after stage $s$. Then $R_d$ has only finitely many incremental stages.
	\end{claim}
	\begin{proof}
		The restraint $r_d$ is defined when $R_d$ is activated, and never changes after stage $s$. Suppose to the contrary that there are incremental stages $s_0 < s_1 < s_2 < \cdots$ after stage $s$. We claim that $\gamma_{s_{i+1}} \geq \frac{1}{2}r_d + \gamma_{s_i}$. From this it follows that there are at most $2/r_d$ incremental stages for $R_d$, as if there were that many incremental stages, for some sufficiently large stage $t$ we would have $\gamma_t$ greater than $1$ and hence greater than $\alpha_{t-1}$---and so the construction could immediately end, with finitely many incremental stages.
		
		Fix $i$ for which we will show that $\gamma_{s_{i+1}} \geq \frac{1}{2}r_d + \gamma_{s_i}$. Since stage $s_i$ is incremental, at the start of that stage $R_d$ is in stage \textbf{restraining}. There are two cases, depending on whether $\beta_{s_i} - \gamma_{s_i} < r_d$ or $\beta_{s_i} - \gamma_{s_i} \geq r_d$.	
		
		\textbf{Case 1:} $\beta_{s_i} - \gamma_{s_i} < r_d$. During stage $s_i$, the requirement $R_d$ enters state \textbf{waiting}. Since stage $s_{i+1}$ is the next incremental stage, there must be some unique stage $t$, $s_i < t < s_{i+1}$, where $R_d$ enters state \textbf{restraining} again and stays in that state until at least stage $s_{i+1}$. At stage $t$ we define $l_d = \alpha_{t-1}$ and $q_d = r_d - (\alpha_{t-1} - \gamma_t)$. These values cannot be redefined until the next incremental stage, $s_{i+1}$, where we have $\gamma_{s_{i+1}} > l_d + \frac{1}{2} q_d$. Then:
		\begin{align*}
		\gamma_{s_{i+1}} &> l_d + \frac{1}{2} q_d \\
		& = \alpha_{t-1} + \frac{1}{2}(r_d - (\alpha_{t-1} - \gamma_{t}))\\
		& = \frac{1}{2}r_d + \frac{1}{2}(\alpha_{t-1} +\gamma_{t})\\
		& \geq \frac{1}{2}r_d + \gamma_{t} \\
		& \geq \frac{1}{2}r_d + \gamma_{s_i}.
		\end{align*}
		
		\textbf{Case 2:} $\beta_{s_i} - \gamma_{s_i} \geq r_d$. During stage $s_i$, the requirement $R_d$ remains in state \textbf{restraining}, defining $\ell_d = \alpha_{s_i - 1}$ and $q_d = r_d - (\alpha_{s_{i-1}}-\gamma_{s_i - 1})$. It stays in that state, with the same reference values $q_d$ and $l_d$, until the next incremental stage $s_{i+1}$, where we have $\gamma_{s_{i+1}} > l_d + \frac{1}{2} q_d$. We get a similar computation to the previous case:
		\begin{align*}
		\gamma_{s_{i+1}} &> l_d + \frac{1}{2} q_d \\
		& = \alpha_{s_i-1} + \frac{1}{2}(r_d - (\alpha_{s_i-1} - \gamma_{s_i}))\\
		& = \frac{1}{2}r_d + \frac{1}{2}(\alpha_{s_i-1} +\gamma_{s_i})\\
		& \geq \frac{1}{2}r_d + \gamma_{s_i}.
		\end{align*}
		
		Thus for each $i$ we have $\gamma_{s_{i+1}} \geq \frac{1}{2}r_d + \gamma_{s_i}$, completing the proof of the claim.
	\end{proof}
	
	\begin{claim}\label{claim:fin-inj}
		Each requirement is injured only finitely many times.
	\end{claim}
	\begin{proof}
		We argue by induction on the priority of the requirements. Suppose that each requirement of higher priority than $R_d$ is only injured finitely many times. Fix a stage $s$ after which none of them are injured. By the previous claim, by increasing $s$ we may assume that no higher priority requirement has an incremental stage after stage $s$.
		
		First of all, $R_d$ can only be activated at stages when every higher priority requirement is \textbf{waiting}. If $R_d$ is never activated after stage $s$, then it cannot be injured. Increasing $s$ further, assume that $R_d$ is activated at stage $s$. If $R_d$ is injured after stage $s$, it is at the first stage $t > s$ such that a requirement $R_e$ of higher priority than $R_d$ has $\beta_t - \gamma_t > r_e$. Moreover, the requirement $R_e$ remains in the state \textbf{waiting} until such a stage. Suppose that $t > s$ is the first such stage, if one exists. At the beginning of stage $t$, $R_e$ is \textbf{waiting}, and so $R_e$ enters the state \textbf{restraining}. Then $R_e$ can only leave state \textbf{restraining}, and re-enter state \textbf{waiting}, at a stage which is incremental for $R_e$; since there are no such stages after stage $s$, $R_e$ can never re-enter stage \textbf{waiting}. So even $R_d$ is never again re-activated, and so cannot be injured. Thus $R_d$ can be injured only once after stage $s$, proving the claim.
	\end{proof}
	
	\begin{claim}\label{claim:random}
		$\alpha = \lim_s \alpha_s$ is random.
	\end{claim}
	\begin{proof}
		There are three possibilities.
		\begin{enumerate}
			\item Some requirement enters state \textbf{preparing} at stage $s$, and is never injured nor leaves state \textbf{preparing} after stage $s$.
			
			\medskip{}
			
			\noindent The requirement $R_d$ is the lowest priority requirement which is active at any point after stage $s$. In this case, at each stage $t \geq s$, we set $\alpha_{t} = \beta_t$ and so $\alpha = \beta$ is random.
			
			\item Some requirement enters state \textbf{restraining} at stage $s$, and is never injured nor leaves state \textbf{restraining} after stage $s$.
			
			\medskip{}
			
			\noindent The requirement $R_d$ is the lowest priority requirement which is active at any point after stage $s$. Increasing $s$, we may assume that this requirement $R_d$ never has an incremental stage after stage $s$. Then the target value $q_d \beta + l_d$ at stage $s$ is also the target value at all stages $t \geq s$. At each such stage $t \geq s$, we set $\alpha_{t+1} = q_d \beta_t + l_d$. Thus $\alpha = q_d \beta + l_d$, with $q_d,l_d \in \mathbb{Q}$, and so is random.
			
			\item For each requirement there is a stage $s$ after which the requirement is never injured and is always in state \textbf{waiting}.
			
			\medskip{}
			
			\noindent There are infinitely many stages $s$ at which we are in Case 2 of the construction. At every stage, all requirements except possibly for the lowest priority requirement are in state \textbf{waiting}. For requirements $R_1,\ldots,R_n$, there is some first stage $t$ at which the lowest priority requirement is in state \textbf{waiting} and never again leaves state waiting. At stage $t$, we must be in Case 2 of the construction. Indeed, in Case 3 the requirement $R_d$ leaves state \textbf{waiting}. In Case 2, we set $\alpha_{t} = \beta_t$. Moreover, we activate the next requirement, and the next requirement is never injured. So there is a greater corresponding first stage $t$ at which that requirement is in state \textbf{waiting} and never again leaves that state. Continuing, there are infinitely many stages at which we set $\alpha_{t} = \beta_t$. 
			It follows that $\alpha = \beta$, which is random. \qedhere
		\end{enumerate}
	\end{proof}
	
	\begin{claim}
		Suppose that $\mu(\dom(M)) = \alpha$. For each requirement $R_d$, there is a stage $s$ after which the requirement is active, never injured, and is always in state \textbf{waiting}.
	\end{claim}
	\begin{proof}
		We argue inductively that for each requirement $R_d$, there is a stage $s$ after which the requirement is never injured and is always \textbf{waiting}.
		
		By Claim \ref{claim:fin-inj} there is a stage $s$ after which $R_d$ is never injured, and (inductively) every higher priority requirement is always \textbf{waiting} after stage $s$. By Claim \ref{claim:fin-exp}, by increasing $s$ we may assume that $R_d$ has no incremental stages after stage $s$.
		
		Then $R_d$ is activated at the least such stage $s$ since each higher priority requirement is always \textbf{waiting}. Note that $R_d$ can never be injured after stage $s$, as if $R_d$ is injured by $R_e$, then $R_e$ enters state \textbf{restraining}.
		
		Now we claim that, if $R_d$ is \textbf{preparing}, it leaves that state after stage $s$. Indeed, if $R_d$ never left state preparing, we would have $\alpha = \beta$. By assumption, $\alpha = \mu(\dom(M)) = \lim_s \gamma_s$. Thus for some stage $t$ we must have that $\beta_t - \gamma_t < r_d$. At this stage $t$, $R_d$ leaves state \textbf{preparing}.
		
		Now we claim that $R_d$ can never enter state \textbf{restraining} after stage $s$. Since $R_d$ has no incremental stages after stage $s$, if $R_d$ did enter state \textbf{restraining}, it would never be able to leave that state. Moreover, $q_d$ and $l_d$ can never change their values. So we end up with $\alpha = q_d \beta + l_d$. Moreover, for all $t \geq s$, $\gamma_t < l_d + \frac{1}{2} q_d$, as there are no more incremental stages. Then $\gamma \leq l_d + \frac{1}{2}q_d < l_d + q_d \beta = \alpha$, contradicting the hypotheses of the claim. Thus $R_d$ can never enter state \textbf{restraining} after stage $s$.
		
		Thus we have shown that for sufficiently large stages, $R_d$ is in state \textbf{waiting}.
	\end{proof}
	
	\begin{claim}\label{claim:sat}
		Suppose that $\mu(\dom(M)) = \alpha$. Then every requirement $R_d$ is satisfied.
	\end{claim}
	\begin{proof}
		Since $\mu(\dom(M)) = \alpha$, at all stages $s$, $\gamma_s \leq \alpha_{s-1}$. As argued in the previous claim, there is a stage $s$ at which $R_d$ is activated, and after which $R_d$ is never injured. At this stage $s$, $R_d$ enters state \textbf{preparing} and we choose $\tau_d$ such that $Q(\tau_d) \uparrow$ and set $r_d = 2^{-(|\tau_d| + d)}$.
		
		By the previous claim, $R_d$ exits state \textbf{preparing} at some stage $t > s$. At this point, we have $\beta_t - \gamma_t < r_d$. We choose a string $\sigma$ such that $K_M(\sigma) > |\tau_d| + d$ and put $Q(\tau_d) = \sigma$. Thus $K_Q(\sigma) \leq |\tau_d|$. $R_d$ enters state waiting, and $\alpha_s = \beta_s$.
		
		Since, at stage $t$, $K_M(\sigma) > |\tau|_d + d$, for every string $\rho$ with $|\rho| \leq |\tau_d| + d$, $M(\rho) \neq \sigma$.  For each stage $t' \geq t$ the requirement $R_d$ is no longer in state \textbf{preparing} and so by Claim \ref{claim:alpha-well-defined} we have $\gamma_{t' + 1} - \gamma_{t'} \leq \alpha_{t'} - \gamma_{t'} < r_d$. From this it follows that we can never have $M(\rho) = \sigma$ for any $\rho$ with $|\rho| \leq |\tau_d| + d$; if $M(\rho) = \sigma$ for the first time at stage $t' + 1 > t$, then we would have $\gamma_{t'+1} - \gamma_{t'} \geq |\rho| = r_d$, which as we just argued cannot happen.
	\end{proof}
	
	We can now use the claims to complete the verification. By Claim \ref{claim:random}, $\alpha = \lim_s \alpha_s$ is indeed random, and by Claim \ref{cl:1} $\alpha \leq \beta$ and so $\alpha \in [0,1]$. So the function $f$ must output the index of a machine $M$ with $\mu(M) = \alpha$. By Claim \ref{claim:sat}, each requirement is satisfied and so, for every $d$, there is $\sigma$ such that $K_{M}(\sigma) > K_Q(\sigma) + d$. Thus $M$ is not optimal, a contradiction. This completes the proof of the theorem.
\end{proof}

\subsection{Almost uniform constructions of optimal machines}

We just established that there is no uniform procedure to turn a left-c.e.\ Martin-L\"of random $\alpha \in [0,1]$ into a universal machine~$M$ such that $\Omega_M=\alpha$. However, algorithmic randomness offers a notion of `almost uniformity', known as layerwise computability, see~\cite{HoyrupR2009}: Let $(\mathcal{U}_k)$ be a fixed effectively optimal Martin-L\"of test, i.e., a Martin-L\"of test such that for any other Martin-L\"of test $(\mathcal{V}_k)$, there exists a constant~$c$ such that $\mathcal{V}_{k+c} \subseteq \mathcal{U}_k$ for all~$k$, and this constant~$c$ can be uniformly computed in an index of the Martin-L\"of test $(\mathcal{V}_k)$. Note that an effectively optimal Martin-L\"of test is in particular universal, i.e., $x$ is Martin-L\"of random if and only if $x \notin \mathcal{U}_d$ for some~$d$. A function $F$ from $[0,1]$ (or more generally, from a computable metric space) to some represented space $\mathcal{X}$ is layerwise computable if it is defined on every Martin-L\"of random~$x$ and moreover there is a partial computable~$f$ from $[0,1] \times \N$ to $\mathcal{X}$ where $f(x,d)=F(x)$ whenever $x \notin \mathcal{U}_d$.

Here we are in a different setting as we are dealing with indices of reals instead of reals, but by extension we could say that a partial function $F: \N \rightarrow \mathcal{X}$ is layerwise computable on left-c.e.\ reals if $F(e)$ is defined for every index~$e$ of a  random left-c.e.\ real, and if there is a partial computable function $f: \N \times \N \rightarrow \mathcal{X}$ such that $f(e,d)=F(e)$ whenever the left-c.e.\ real $\alpha_e$ of index~$e$ does not belong to $\mathcal{U}_d$ (note that the definition remains the same if~$f$ is required to be total). Even with this weaker notion of uniformity, uniform construction of optimal machines from their halting probabilities remains impossible.

\begin{theorem}
There does not exist a layerwise computable mapping~$F$ from indices for random left-c.e.\ reals $\alpha_e \in [0,1]$ to optimal machines such that $\Omega_{M_{F(e)}}=\alpha_e$.
\end{theorem}

\begin{proof}
This is in fact a consequence of a stronger result: there is no $\emptyset'$-partial computable function $F$ such that $F(e)$ is defined whenever $\alpha_e$ is Martin-L\"of random and $\Omega_{M_{F(e)}}=\alpha_e$. Since a $\emptyset'$-partial computable function can be represented by a total computable function $f(.,.)$ such that for every $e$ on which $F$ is defined, $\lim_t f(e,t)=F(e)$, we see that a layerwise computable function on left-c.e.\ reals is a particular case of $\emptyset'$-partial computable function.\\

Let now $F$ be a $\emptyset'$-partial computable function and $f$ a total computable such that $\lim_t f(e,t) = F(e)$ whenever~$F(e)$ is defined.

The idea is to run the same construction as in Theorem~\ref{no-uniform-machine-constr}, but instead of playing against the machine of index $f(e)$, we play against the machine of index $f(e,s_0)$, with $s_0=0$. If at some point we find a $s_1>s_0$ such that $f(e,s_1) \not= f(e,s_0)$, we restart the entire construction, this time playing against the machine of index $f(e,s_1)$, until we find $s_2>s_1$ such that  $f(e,s_2) \not= f(e,s_1)$, then restart, etc. Of course when we restart the construction, we cannot undo the increases we have already made on $\alpha$. This problem is easily overcome as follows. First observe that the strategy presented in the proof of Theorem~\ref{no-uniform-machine-constr}, is robust: instead of starting at $\alpha=0$, and staying in the interval $[0,1]$ throughout the construction, for any rational interval $[a,b] \subseteq [0,1]$, we could have started the construction with $\alpha_0=a$ and stayed within $[a,b]$ by -- for example -- targeting the random real $a+(b-a)\beta$ instead of $\beta$. Now, let $\xi$ be a random left-c.e.\ real in $[0,1]$ with computable lower approximation $\xi_0 < \xi_1 < \ldots$. We play against the machine of index $f(e,s_i)$ by applying the strategy of Theorem~\ref{no-uniform-machine-constr} with the added constraint that $\alpha$ must stay in the interval $[\xi_i,\xi_{i+1}]$. If we then find a $s_{i+1}$ such that $f(e,s_{i+1}) \not= f(e,s_i)$, we then move to the next interval $[\xi_{i+1},\xi_{i+2}]$ and apply the strategy to diagonalize against the machine of index $f(e,s_{i+1})$ while keeping~$\alpha$ in this interval, etc.

There are two cases:
\begin{itemize}
\item Either $f(e,t)$ eventually stabilizes to a value $f(e,s_k)$, in which case we get to fully implement the diagonalization against the machine of index $f(e,s_k)=F(e)$, which ensures that $\alpha_e \not= \Omega_{M_{F(e)}}$ or that $M_{F(e)}$ is not optimal.
\item Or $f(e,t)$ does not stabilize, in which case we will infinitely often move $\alpha$ from the interval $[\xi_i,\xi_{i+1}]$ to $[\xi_{i+1}, \xi_{i+2}]$, which means that the limit value of $\alpha=\alpha_e$ will be $\xi$, hence $\alpha_e$ is random, while $F(e)$ is undefined since $f(e,t)$ does not converge.
\end{itemize}

In either case, we have shown what we wanted.
\end{proof}

Finally, we can consider a yet weaker type of non-uniformity. In the definition of layerwise computability on left-c.e.\ reals, we asked that for $\alpha_e \notin \mathcal{U}_d$, the machine of index $f(e,d)$ has halting probability~$\alpha_e$ \emph{and} $f(e,d) = f(e,d')$ if $\alpha_e \notin \mathcal{U}_d \cup \mathcal{U}_{d'}$. Here we could try to remove this last condition by allowing $f(e,d)$ and $f(e,d')$ to be codes for different machines (but both with halting probabilities $\alpha_e$). In this setting, we do get a positive result.

\begin{theorem}
There exists a partial computable function~$f(.,.)$ such that if $\alpha_e \notin \mathcal{U}_d$, $\alpha_e \in [0,1]$, then $f(e,d)$ is defined and $\Omega_{M_{f(e,d)}}=\alpha_e$.
\end{theorem}

\begin{proof}
This follows from work of Calude, Hertling, Khoussainov, and Wong~\cite{CaludeHertlingKhoussainovWang} and of Ku\v cera and Slaman~\cite{KuceraS2001}. Let $\Omega$ be the halting probability of an optimal machine. Ku\v cera and Slaman showed how from the index of a left-c.e.\ real $\alpha \in [0,1]$ one can build a Martin-L\"of test $(\mathcal{V}_k)$ such that if $\alpha \notin (\mathcal{V}_k)$ then one can, uniformly in~$k$, produce approximations $\alpha_1 < \alpha_2 < \ldots$ of $\alpha$ and $\Omega_1 < \Omega_2 < \ldots $ of $\Omega$ such that $(\alpha_{s+1} - \alpha_s) > 2^{-k} (\Omega_{s+1}-\Omega_s)$ (see~\cite[Theorem 9.2.3]{DowneyH2010}). Then, by~\cite{CaludeHertlingKhoussainovWang}, one can use such approximations to uniformly build a uniform machine with halting probability~$\alpha$, as long as $\alpha \in (2^{-k},1-2^{-k})$ (see~\cite[Theorem 9.2.2]{DowneyH2010})

Thus, given an index for $\alpha$, if $(\mathcal{V}_k)$ is the Martin-L\"of test built as in~\cite{KuceraS2001}, we can build the test $\mathcal{V}'_k  = \mathcal{V}_{k+2} \cup (0,2^{-k-2}) \cup (1-2^{-k-2},1)$ (whose index can uniformly be computed from that of $(\mathcal{V}_k)$). Now, if $\alpha \notin \mathcal{U}_d$, then we can compute a constant~$c$ such that $\alpha \notin  \mathcal{V}'_{d+c}$, and apply the above argument with $k=c+d+2$.
\end{proof}

\subsection{Uniform constructions of semi-measures}

Another way to define Omega numbers, which is equivalent if one is not concerned about uniformity issues, is  via left-c.e.\ semi-measures (see~\cite[Section 3.9]{DowneyH2010}).

\begin{definition}
	A semi-measure is a function $m: \mathbb{N} \rightarrow \mathbb{R}^+$ such that $\sum_i m(i) \leq 1$. It is left-c.e.\ if the set $\{(i,q) \mid i \in \mathbb{N}, q \in \mathbb{Q}, \ m(i)>q\}$ is c.e., or equivalently, if $m$ is the limit of a non-decreasing sequence $(m_s)$ of uniformly computable functions such that $\sum_s m_s(i) \leq 1$ for all~$s$.
\end{definition}

There exist universal left-c.e.\ semi-measures, i.e.,  left-c.e.\ semi-measures $m$ such that for any other left-c.e.\ semi-measure $\mu$, there is a $c>0$ such that $m(i)>c \cdot \mu(i)$ for all~$i$. The Levin coding theorem (see~\cite[Theorem 3.9.4]{DowneyH2010}) asserts that a left-c.e.\ semi-measure~$m$ is universal if and only if there are positive constants $c_1,c_2$ such that $c_1 \cdot 2^{-K(i)} < m(i) < c_2 \cdot 2^{-K(i)}$ for all~$i$. An important result from Calude, Hertling, Khoussainov, and Wang~\cite{CaludeHertlingKhoussainovWang} is that a left-c.e.\ real $\alpha$ is an Omega number if and only if it is the sum $\sum_i m(i)$ for some universal left-c.e.\ semi-measure~$m$. Interestingly, with this representation of Omega numbers, uniform constructions are possible.

\begin{theorem}
	There is a total computable function $f$ such that if $e$ is an index for a random left-c.e.\ real $\alpha \in [0,1]$, then $f(e)$ is defined and is an index for a universal left-c.e.\ semi-measure $m_{f(e)}$ with sum $\alpha$.
\end{theorem}

\begin{proof}
	Let $\mu$ be a fixed universal semi-measure and $\gamma \leq 1$ its sum. Suppose we are given (the index of) a left-c.e.\ real $\alpha$. We build our $m$ by building uniformly, for each~$k>0$, a left-c.e.\ semi-measure $m_k$ of halting probability $\alpha \cdot 2^{-k}$ and will take $m = \sum_{k>0} m_k$. While doing so, we also build an auxiliary Martin-L\"of test $(\mathcal{U}_k)_{k>0}$.
	
	The measure $m_k$ is designed as follows. We monitor the semi-measure $\mu$ and $\alpha$ at the same time and run the following algorithm
	\begin{enumerate}
		\item[1.] Let $s_0$ be the stage at which we entered step 1. Wait for the least stage $s \geq s_0$ such that some value $\mu(i)$ with $i \leq s$ has increased since the last $i$-stage. If there is more than one such $i$ at stage $s$, let $i$ be the one whose most recent $i$-stage is least. Let $x$ be the amount by which $\mu(i)$ has increased since the previous $i$-stage, and say that $s$ is an $i$-stage. Move to step 2.
		\item[2.] Put $(\alpha_s,\alpha_s+2^{-k}x)$ into $\mathcal{U}_k$. Move to step 3.
		\item[3.] Increase $m_k(i)$ by $2^{-k}(\alpha_s - \alpha_{s_0})$. At further stages $t \geq s$, when we see an increase $\alpha_{t+1}>\alpha_t$, we increase $m_k(i)$ by $2^{-k} (\alpha_{t+1}-\alpha_t)$. Moreover, if we now have $\alpha_{t+1} > \alpha_s+2^{-k}x$, we go back to step 1, otherwise we stay in this step 3.
	\end{enumerate}
	
	By construction we do have $\sum_i m_k(i) = 2^{-k}\alpha$. Still by construction, the measure of $\mathcal{U}_k$ is bounded by $\gamma \cdot 2^{-k} \leq 2^{-k}$, so it is indeed a Martin-L\"of test. Thus, if $\alpha$ is indeed random, there is a~$j$ such that $\alpha \notin \mathcal{U}_j$. Looking at the above algorithm, $\alpha \notin \mathcal{U}_j$ means that for this~$j$, we enter step~1 of the algorithm infinitely often and thus whenever some $\mu(i)$ is increased by~$x$ at step 1, this is met by a sum of increases of $m_j(i)$ by strictly more than $2^{-j}x$ during step 3. Thus, $m_j > 2^{-j}\mu$, which makes $m_j$ a universal semi-measure, and thus $m>m_j$ is universal.
\end{proof}

An interesting corollary is that one cannot uniformly turn a universal left-c.e.\ semi-measure~$m$ into a prefix-free machine whose halting probability is $\sum_i m(i)$. Indeed, if we could, then we could uniformly turn a random left-c.e.\ $\alpha \in [0,1]$ into a prefix-free machine of halting probability $\alpha$ by first applying the above theorem to get a universal left-c.e.\ semi-measure~$m$ of sum $\alpha$, and then we could turn $m$ into a machine~$M$ of sum $\alpha$. This would contradict Theorem~\ref{no-uniform-machine-constr}.

To summarize, for arbitrary (not necessarily random) left-c.e.\ reals, we can make all of the transformations uniformly:
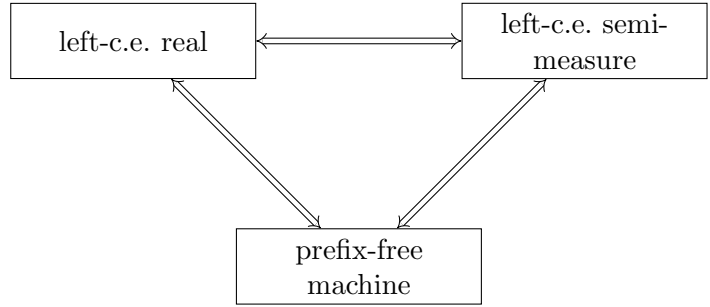
\[
\begin{tikzpicture}
\tikzstyle{every rectangle node}=[draw,  text width=3cm, minimum height = 1cm, align = center]
\node (ce) at (0,0) [rectangle] {left-c.e.\ real};
\node (measure) at (6,0) [rectangle] {left-c.e.\ semi-measure};
\node (machine) at (3,-3) [rectangle] {prefix-free machine};

\draw [implies-implies,double equal sign distance] (ce) -- (measure);
\draw [implies-implies,double equal sign distance] (machine) -- (measure);
\draw [implies-implies,double equal sign distance] (machine) -- (ce);

\end{tikzpicture}\]

For random left-c.e.\ reals, and optimal prefix-free machines, we can only make the following transformations uniformly:
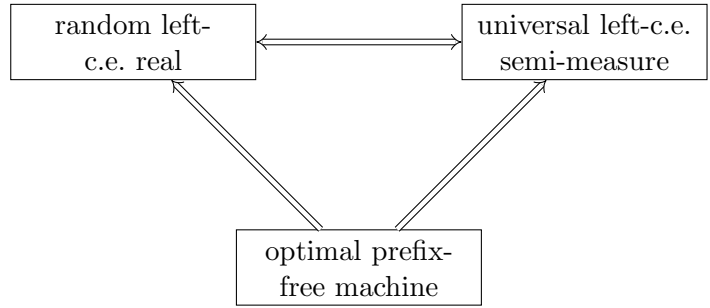
\[
\begin{tikzpicture}
\tikzstyle{every rectangle node}=[draw,  text width=3cm, minimum height = 1cm, align = center]
\node (ce) at (0,0) [rectangle] {random left-c.e.\ real};
\node (measure) at (6,0) [rectangle] {universal left-c.e.\ semi-measure};
\node (machine) at (3,-3) [rectangle] {optimal prefix-free machine};

\draw [implies-implies,double equal sign distance] (ce) -- (measure);
\draw [-implies,double equal sign distance] (machine) -- (measure);
\draw [-implies,double equal sign distance] (machine) -- (ce);

\end{tikzpicture}\]

\section{Differences of Left-c.e.\ Reals}

{
\renewcommand{\thetheorem}{\ref{thm:no-unif-diff}}
\begin{theorem}
	There is no partial computable function $f$ such that if $e$ is an index for a non-computable left-c.e.\ real $\alpha$, then $f(e)$ is defined and is an index for a left-c.e.\ real $\beta$ such that $\alpha - \beta$ is neither left-c.e.\ nor right-c.e.
\end{theorem}
}
\addtocounter{theorem}{-1}
\begin{proof}
Using the recursion theorem, define a left-c.e.\ real $\alpha$ while watching the left-c.e.\ real $\beta$ produced from $\alpha$ by a function $f$. We will also define a right-c.e.\ real $\delta$. Let $\theta^i$ be an enumeration of the right-c.e.\ reals with right-c.e.\ approximations $(\theta^i_s)$. We will ensure that $\alpha \neq \theta^i$ for any $i$, so that $\alpha$ is non-computable, and that either $\alpha - \beta = \delta$ or for all sufficiently large stages, $\alpha$ grows more than $\beta$ (and so $\alpha - \beta$ is left-c.e.).

Each stage of the construction will be in one of infinitely many possible states: \textbf{wait} and \textbf{follow$(i)$} for some $i$. In \textbf{wait}, $\alpha$ will be held to the same value and we will begin decreasing the right-c.e.\ real $\delta$ closer to $\alpha - \beta$; if there are infinitely many \textbf{wait} stages, then in fact we will have $\delta = \alpha - \beta$. At \textbf{follow$(i)$} stages, $\alpha$ will increase as much as $\beta$, and possibly more, in an attempt to have $\theta^i < \alpha$. Because $\alpha$ will be increasing as much as and possibly more than $\beta$, if from some point on all stages are  \textbf{follow$(i)$} stages, then $\alpha - \beta$ will be left-c.e. We will only enter \textbf{follow$(i)$} when we have a reasonable chance of making $\theta^i < \alpha$, i.e., when $\theta^i$ is not too much greater than $\alpha$, and we will only exit \textbf{follow$(i)$} when we have succeeded in making $\theta^i < \alpha$. Since $\theta^i$ is right-c.e.\ and $\alpha$ is left-c.e., this can never be injured. It is possible that we will never succeed in making $\theta^i < \alpha$ (because in fact $\theta^i > \alpha$) but in this case we will still ensure that $\alpha$ is not computable and make $\alpha - \beta$ left-c.e. We just have to make sure that we never increase $\alpha - \beta$ above $\delta$.

\medskip{}

Note that technically when defining $\alpha_s$ we cannot wait for $\beta_s$ to converge. But we can do this by essentially the following argument. First, fix a non-computable left-c.e.\ real $\gamma$ and let $\alpha_s = \gamma_s$ until the uniform procedure provides us with a $\beta$ and $\beta_0$ converges at some stage $s_0$. Then we can restart the construction, considering the construction to begin with $\alpha_0 = \gamma_{s_0}$. We can also in a uniform way replace the given approximation to $\beta$ (which might not even be total or left-c.e.) by a different one which is guaranteed to be left-c.e.\ and which converges in a known amount of time, and is equal to $\beta$ in the case that $\beta$ is in fact left-c.e.

\medskip{}

\noindent \textit{Construction}.

\medskip{}

\noindent Stage $s = 0$. Begin with $\alpha_0 = \gamma_{s_0}$, $\delta_0 = 1 + \alpha_0 - \beta_0$. Say that stage 1 will be a \textbf{wait} stage.

\medskip{}

\noindent Stage $s+1$. We will have determined in stage $s$ whether stage $s+1$ is a \textbf{wait} or \textbf{follow} stage.

\medskip{}

\noindent \textbf{wait}: Let $\alpha_{s+1} = \alpha_s$ and
\[\delta_{s+1} = \min\left(\delta_s,\alpha_{s+1}-\beta_{s+1} + \frac{1}{2^s}\right).\]
Check whether, for some $i \leq s$, $\theta^i_{s+1} \geq \alpha_{s+1}$ and $\theta^i_{s+1} - \alpha_{s+1} < \frac{1}{2^i}$. If we find such an $i$, let $i$ be the least such. The next stage is a \textbf{follow$(i)$} stage. If there is no such $i$, the next stage is a \textbf{wait} stage.

\medskip{}

\noindent \textbf{follow$(i)$}: In all cases, let $\delta_{s+1} = \delta_s$. Then:
\begin{enumerate}
	\item Check whether
\[ \alpha_s + \beta_{s+1} - \beta_s > \theta^i_{s+1}.\]
If so, set
\[ \alpha_{s+1} = \theta^i_{s+1} + \epsilon \leq \alpha_s + \beta_{s+1} - \beta_s \]
where $\epsilon < \frac{1}{2^i}$. The next stage is a \textbf{wait} stage.

\item Otherwise, check whether for some $j$,
\[ 0 \leq \theta^j_s - \alpha_s < \frac{1}{2^{j+2}} \Big[ \delta_s - (\alpha_s-\beta_s) \Big]. \]
If we find such a $j$, choose the least such $j$, and let $\epsilon > 0$ be such that
\[ \theta^j_s + \epsilon - \alpha_s < \frac{1}{2^{j+2}} \Big[ \delta_s - (\alpha_s-\beta_s) \Big]. \]
Let
\[\alpha_{s+1} = \max(\theta^j_s + \epsilon,\alpha_s + \beta_{s+1}-\beta_s).\]
If $j = i$, the next stage is a \textbf{wait} stage. Otherwise, the next stage is a \textbf{follow$(i)$} stage.

\item Finally, in any other case, let
\[\alpha_{s+1} = \alpha_s + \beta_{s+1}-\beta_s.\]
The next stage is a \textbf{follow$(i)$} stage.
\end{enumerate}

\medskip{}

\noindent \textit{End construction.}

\medskip{}

The verification will consist of five claims followed by a short argument.

\begin{claim}\label{cl:1}
	$\alpha = \sup \alpha_s$ comes to a limit.
\end{claim}
\begin{proof}
	In \textbf{wait} stages, we do not increase $\alpha$. If we enter \textbf{follow$(i)$}, then we can increase $\alpha$ by at most $\frac{2}{2^i}$ before we exit \textbf{follow$(i)$}. Thus
	\[ \alpha \leq \sum_{i \in \omega} \frac{2}{2^i} < \infty.\qedhere\]
\end{proof}

\begin{claim}\label{cl:2}
	Suppose that, from some stage $t$ on, every stage is a \textbf{follow$(i)$} stage. Then:
	\begin{enumerate}
		\item for all $s \geq t$, $\alpha_{s+1} - \alpha_s \geq \beta_{s+1} - \beta_s$,
		\item $\alpha - \beta$ is left-c.e.,
		\item for all $s \geq t$,
		\[\delta_s - (\alpha_s-\beta_s) \geq \frac{1}{2}\Big[\delta_t - (\alpha_t-\beta_t)\Big].\]
	\end{enumerate}
\end{claim}
\begin{proof}
	(1) follows from the fact that we either set the next stage to be a \textbf{wait} stage, or we have $\alpha_{s+1} \geq \alpha_s + \beta_{s+1} - \beta_s$. (2) follows easily from (1).
	
	For (3), since $\delta_s = \delta_t$ for all $s \geq t$, whenever we define
	\[\alpha_{s+1} = \alpha_s + \beta_{s+1}-\beta_s \]
	we maintain
	\[ \delta_{s+1} - (\alpha_{s+1}-\beta_{s+1}) = \delta_s - (\alpha_s-\beta_s).\]
	The other possible case is when we find $j$ such that
	\[ 0 \leq \theta^j_s - \alpha_s < \frac{1}{2^{j+2}} \Big[ \delta_s - (\alpha_s-\beta_s) \Big]. \]
	and define
	\[\alpha_{s+1} = \theta^j_s + \epsilon.\]
	Note that in this case we permanently have $\theta^j - \alpha < 0$ so we can never do this again for the same $j$. We have
	\begin{align*}
	\delta_{s+1} - (\alpha_{s+1}-\beta_{s+1}) &= \delta_s - \theta^j_s - \epsilon + \beta_{s+1}\\
	&\geq \delta_s - \alpha_s + \beta_s - \frac{1}{2^{j+2}}\Big[ \delta_s - (\alpha_s-\beta_s) \Big]\\
	&= \frac{2^{j+2} - 1}{2^{j+2}} \Big[ \delta_s - (\alpha_s-\beta_s) \Big].
	\end{align*}
	Thus, for all stages $s \geq t$,
	\begin{align*}
		\delta_s - (\alpha_s-\beta_s) &\geq \prod_{j \in \omega} \frac{2^{j+2}-1}{2^{j+2}} \Big[\delta_t - (\alpha_t-\beta_t)\Big] \\
		&\geq	\frac{1}{2} \Big[\delta_t - (\alpha_t-\beta_t)\Big].\qedhere
	\end{align*}
\end{proof}

\begin{claim}\label{cl:4}
	For all stages $s$, $\delta_s > \alpha_s - \beta_s$.
\end{claim}
\begin{proof}
	We argue by induction. This is true for $s = 0$.
	
	If stage $s+1$ is a \textbf{wait} stage, then there are two possible values for $\delta_{s+1}$: $\delta_s$ or $\alpha_{s+1} - \beta_{s+1} + \frac{1}{2^s}$. It is clear that the second is strictly greater than $\alpha_{s+1} - \beta_{s+1}$. We also have, since $\alpha_{s+1}=\alpha_s$ and $\beta_{s+1} \geq \beta_s$, that $\delta_s > \alpha_s - \beta_s \geq \alpha_{s+1} - \beta_{s+1}$.
	
	If stage $s+1$ is a \textbf{follow} stage, then $\delta_{s+1} = \delta_s$. There are two options for $\alpha_{s+1}$. First, we might set $\alpha_{s+1} \leq \alpha_s + \beta_{s+1} - \beta_s$ so that $\alpha_{s+1} - \beta_{s+1} \leq \alpha_s - \beta_s$ and $\delta_{s+1} > \alpha_{s+1} - \beta_{s+1}$ follows from the induction hypothesis $\delta_s > \alpha_s-\beta_s$. Second, we might set
	\[\alpha_{s+1} = \theta^j_s + \epsilon.\]
	where
	\[ \theta^j_s + \epsilon - \alpha_s < \frac{1}{2^{j+2}} \Big[ \delta_s - (\alpha_s-\beta_s) \Big].\]
	Then
	\begin{align*}
	\alpha_{s+1} - \beta_{s+1} &\leq \theta^j_s + \epsilon - \beta_s\\
	 &< \alpha_s - \beta_s + \frac{1}{2^{j+2}} \Big[ \delta_s - (\alpha_s-\beta_s) \Big]\\
	 &\leq \frac{1}{2^{j+2}} \delta_s + \frac{2^{j+2} - 1}{2^{j+2}} \Big[\alpha_s-\beta_s \Big]\\
	 &<  \delta_s = \delta_{s+1}.
	\end{align*}
	This completes the proof.
\end{proof}

\begin{claim}\label{cl:3}
	$\alpha$ is non-computable.
\end{claim}
\begin{proof}
	If $\alpha$ was computable, then it would be equal to a right-c.e.\ real $\theta^i$. For all stages $s$, $\alpha \leq \theta^i_s$. Let $t$ be a stage such that $\theta^i_t - \alpha_t < \frac{1}{2^i}$. Increasing $t$, we may assume that there is $j \leq i$ such that we are in \textbf{follow$(j)$} from stage $t$ on. Increasing $t$ further, we can assume that for each $i' < i$, if $\theta^i < \alpha$, then we have seen this by stage $t$. Consider the inequality
	\[ \theta^i_s - \alpha_s < \frac{1}{2^{i+2}} \Big[ \delta_s - (\alpha_s-\beta_s) \Big]. \]
	By (3) of Claim \ref{cl:2}, the right-hand-side has a lower bound, and this lower bound is strictly positive by Claim \ref{cl:4}. Since $\theta^i = \alpha$, there is a stage $s \geq t$ where this inequality holds. Then by choice of $t$, $i$ is the least value satisfying this inequality and we set $\alpha_{s+1} > \theta^i_s$.
\end{proof}

\begin{claim}\label{cl:5}
	If there are infinitely many \textbf{wait} stages, then $\delta = \alpha - \beta$.
\end{claim}
\begin{proof}
	Using Claim \ref{cl:4}, for each \textbf{wait} stage $s$, we have
	\[ \alpha_s - \beta_s \leq \delta_s \leq \alpha_s - \beta_s + \frac{1}{2^{s-1}}. \]
	Thus $\delta = \alpha - \beta$.
\end{proof}

We are now ready to complete the proof. It follows from Claim \ref{cl:1} that $\alpha$ is a left-c.e.\ real that comes to a limit, and by Claim \ref{cl:3}, $\alpha$ is a non-computable. If there are infinitely many \textbf{wait} stages, then by Claim \ref{cl:5} $\delta = \alpha - \beta$ is right-c.e. The other option is that there is $j$ such that every stage from some point on is a \textbf{follow$(j)$} stage. In this case, by (2) of Claim \ref{cl:2}, $\alpha - \beta$ is left-c.e.
\end{proof}

We now turn to Theorem \ref{thm:unif-diff-mach} which says that one can uniformly construct, from an optimal (respectively universal) machine $U$, an optimal (respectively universal) machine $V$ such that $\Omega_{U}-\Omega_{V}$ is neither left-c.e.\ nor right-c.e. We first prove this for optimal machines, and then obtain the result for universal machines as a corollary.

\begin{theorem}
	Theorem~\ref{thm:bl2} is uniform, in the sense that there is a total computable function~$f$ such that if $U=M_e$ is an optimal machine, then $V=M_{f(e)}$ is optimal and $\Omega_{U}-\Omega_{V}$ is neither left-c.e.\ nor right-c.e.
\end{theorem}

\begin{proof}
	Let $\gamma, \delta$ be two Solovay-incomparable left-c.e.\ reals. As explained in  \cite{BarmpaliasLewis}, if $\alpha$ is random, then $\beta = \alpha + \gamma - \delta$ is left-c.e.\ and random, and $\alpha-\beta$ is neither left-c.e.\ nor right-c.e.\ Our goal is to make this idea effective.
	
	Let us first express $\delta$ as the sum $\sum_n 2^{-h(n)}$ where $h$ is a computable function. In what follows, when we write $h(\sigma)$ for a string $\sigma$, we mean $h(n)$ where $n$ is the integer associated to $\sigma$ via a fixed computable bijection. Furthermore, let $Q$ be a machine such that $\mu(\dom(Q))=\gamma$.
	
	We build a machine~$V$ from a machine~$U$ as follows. First, we wait for $U$ to issue a description $U(\sigma_0)=\tau_0$. When this happens, $V$ issues a description $V(\sigma_0 0)=\tau_0$ and countably many descriptions by setting $V(\sigma_0 1 p)=Q(p)$ for every $p \in \dom(Q)$.
	
	Now, for every string $\tau \not= \tau_0$ in parallel, we enumerate all descriptions ${U(\sigma)=\tau}$. As long as the enumerated descriptions are such that $|\sigma| \geq h(\tau)$, $V$ copies these descriptions. If at some point we find a description $U(\sigma) = \tau$ with $|\sigma| \leq h(\tau)-1$, we then issue descriptions $V(\sigma0)=\tau$, and $V(\sigma')=\tau$ for every $\sigma'$ of length $h(\tau)$ which extends $\sigma 1$, except for $\sigma'=\sigma 1^{h(\tau)-|\sigma|}$, for which we leave $V(\sigma')$ undefined. After having done that, $V$ copies all further $U$-descriptions of~$\tau$, regardless of the length of these descriptions.
	
	By construction, $V$ is prefix-free, because any $U$-description $U(\sigma)=\tau$ is replaced in $V$ by a set of descriptions $V(\sigma')=\tau'$ where the $\sigma'$ form a prefix-free set of extensions of~$\sigma$. Moreover, $V$ is optimal because by construction, whenever a description $U(\sigma)=\tau$ is enumerated, a $V$-description of $\tau$ of length at most $|\sigma|+1$ is issued. Let us now evaluate $\Omega_U-\Omega_V$. The very first description $U(\sigma_0)=\tau_0$ of $U$ gives rise to descriptions in~$V$ of total measure $2^{-c-1}+2^{-c-1}\mu(\dom(Q))$, where $c=|\sigma_0|$. Thus this part of the construction contributes to $\Omega_U-\Omega_V$ by an amount $2^{-c}-2^{-c-1}-2^{-c-1}\mu(\dom(Q))=2^{-c-1}-2^{-c-1}\gamma$.
	
	Now, for other strings $\tau \not= \tau_0$, there are two cases. Either a description $U(\sigma)= \tau$ with $|\sigma|<h(\tau)$ is found (which is equivalent to saying that $K_U(\tau) < h(\tau)$), or no such description is found. Let $A$ be the set of $\tau$ for which such a description is found. For $\tau \notin A$, all $U$-descriptions of $\tau$ are copied identically in~$V$. For $\tau \in A$, all $U$-descriptions of~$\tau$ are copied except one description $U(\sigma)=\tau$ (thus of measure $2^{-|\sigma|}$) which is mimicked in $V$ by a set of descriptions of measure $2^{-|\sigma|}-2^{-h(\tau)}$.
	
	Putting it all together:
	\[
	\Omega_U-\Omega_V = 2^{-c-1}-2^{-c-1}\gamma + \sum_{\tau \in A} 2^{-h(\tau)}
	\]
	
	To finish the proof, we appeal to the theory of Solovay functions. When $h$ is a computable positive function, the sum $\sum_n 2^{-h(n)}$ is \emph{not} random if and only if $h(n) - K(n) \rightarrow \infty$~\cite{BienvenuD2009,BienvenuDNM2015}. This is the case here as $\delta=\sum_n 2^{-h(n)}$ is Solovay-incomplete hence not random. Suppose that the machine $U$ is indeed an optimal machine. Then $K_U = K +O(1)$, and thus we have $h(n) - K_U(n) \rightarrow \infty$. In particular, for almost all $n$, $h(n) > K_U(n)$. This shows that the set $A$ above is cofinite and therefore that $\sum_{\tau \in A} 2^{-h(\tau)}=\delta-q$ for some (dyadic) rational~$q$. Plugging this in the above equality, we get
	\[
	\Omega_U-\Omega_V = 2^{-c-1}-2^{-c-1}\gamma + \delta - q
	\]
	Since $\gamma$ and $\delta$ are Solovay-incomparable, this shows that $\Omega_U-\Omega_V$ is neither left-c.e.\ nor right-c.e.
\end{proof}

\begin{corollary}
	There is a total computable function~$g$ such that if $U=M_e$ is a universal machine, then $W=M_{g(e)}$ is universal and $\Omega_{U}-\Omega_{W}$ is neither left-c.e.\ nor right-c.e.
\end{corollary}
\begin{proof}
	Given $U = M_e$, construct $V = M_{f(e)}$ as in the previous theorem. Define a machine $W = M_{g(e)}$ by setting $W(0 \sigma) = U(0 \sigma)$ and $W(1 \sigma) = V(1 \sigma)$. Then $\Omega_W = \frac{1}{2} \Omega_U + \frac{1}{2} \Omega_V$, and so $\Omega_U - \Omega_W = \frac{1}{2}(\Omega_U - \Omega_V)$. Thus if $U$ is universal, then so is $W$, and $\Omega_U - \Omega_W$ is neither left-c.e.\ nor right-c.e.
\end{proof}

\bibliography{References}

\begin{thebibliography}{CHKW01}

\bibitem[ASWZ00]{AWZ}
Klaus Ambos-Spies, Klaus Weihrauch, and Xizhong Zheng.
\newblock Weakly computable real numbers.
\newblock {\em J. Complexity}, 16(4):676--690, 2000.

\bibitem[Bar18]{Barmpalias}
George Barmpalias.
\newblock Aspects of {C}haitin's {O}mega.
\newblock In {\em Algorithmic Randomness: Progress and Prospects}, pages
  623--632. Springer, 2018.

\bibitem[BD09]{BienvenuD2009}
Laurent Bienvenu and Rodney Downey.
\newblock {K}olmogorov complexity and {S}olovay functions.
\newblock In {\em Symposium on Theoretical Aspects of Computer Science (STACS
  2009)}, volume 09001 of {\em Dagstuhl Seminar Proceedings}, pages 147--158,
  http://drops.dagstuhl.de/opus/volltexte/2009/1810, 2009. Schloss Dagstuhl -
  Leibniz-Zentrum fuer Informatik, Germany Internationales Begegnungs- und
  Forschungszentrum fuer Informatik (IBFI), Schloss Dagstuhl, Germany.

\bibitem[BDNM15]{BienvenuDNM2015}
Laurent Bienvenu, Rodney~G. Downey, Andr{\'{e}} Nies, and Wolfgang Merkle.
\newblock Solovay functions and their applications in algorithmic randomness.
\newblock {\em Journal of Computer and System Sciences}, 81(8):1575--1591,
  2015.

\bibitem[BLP17]{BarmpaliasLewis}
George Barmpalias and Andrew Lewis-Pye.
\newblock A note on the differences of computably enumerable reals.
\newblock In {\em Computability and complexity}, volume 10010 of {\em Lecture
  Notes in Comput. Sci.}, pages 623--632. Springer, Cham, 2017.

\bibitem[BS12]{BienvenuS2012}
Laurent Bienvenu and Alexander Shen.
\newblock Random semicomputable reals revisited.
\newblock In Michael~J. Dinneen, Bakhadyr Khoussainov, and Andr{\'e} Nies,
  editors, {\em Computation, Physics and Beyond}, volume 7160 of {\em Lecture
  Notes in Computer Science}, pages 31--45. Springer, 2012.

\bibitem[Cha75]{Chaitin75}
Gregory~J. Chaitin.
\newblock A theory of program size formally identical to information theory.
\newblock {\em J. Assoc. Comput. Mach.}, 22:329--340, 1975.

\bibitem[CHKW01]{CaludeHertlingKhoussainovWang}
Cristian~S. Calude, Peter~H. Hertling, Bakhadyr Khoussainov, and Yongge Wang.
\newblock Recursively enumerable reals and {C}haitin {$\Omega$} numbers.
\newblock {\em Theoret. Comput. Sci.}, 255(1-2):125--149, 2001.

\bibitem[CNSS11]{CNSS}
Cristian~S. Calude, Andr\'{e} Nies, Ludwig Staiger, and Frank Stephan.
\newblock Universal recursively enumerable sets of strings.
\newblock {\em Theoret. Comput. Sci.}, 412(22):2253--2261, 2011.

\bibitem[CS09]{CaludeStaiger}
Cristian~S. Calude and Ludwig Staiger.
\newblock On universal computably enumerable prefix codes.
\newblock {\em Math. Structures Comput. Sci.}, 19(1):45--57, 2009.

\bibitem[DH10]{DowneyH2010}
Rod Downey and Denis Hirschfeldt.
\newblock {\em Algorithmic Randomness and Complexity}.
\newblock Theory and Applications of Computability. Springer, 2010.

\bibitem[DHN02]{DowneyHirschfeldtNies}
Rod Downey, Denis~R. Hirschfeldt, and Andr\'{e} Nies.
\newblock Randomness, computability, and density.
\newblock {\em SIAM J. Comput.}, 31(4):1169--1183, 2002.

\bibitem[HR09]{HoyrupR2009}
Mathieu Hoyrup and Crist{\'o}bal Rojas.
\newblock Computability of probability measures and {M}artin-{L\"o}f randomness
  over metric spaces.
\newblock {\em Information and Computation}, 207(7):2207--2222, 2009.

\bibitem[KS01]{KuceraS2001}
Antonin Ku{\v c}era and Ted Slaman.
\newblock Randomness and recursive enumerability.
\newblock {\em SIAM Journal on Computing}, 31:199--211, 2001.

\bibitem[Ng06]{ng}
Keng~Meng Ng.
\newblock Some properties of d.c.e. reals and their degrees.
\newblock Master's thesis, National University of Singapore, 2006.

\bibitem[Nie09]{Nies2009}
Andr{\'e} Nies.
\newblock {\em Computability and Randomness}.
\newblock Oxford Logic Guides. Oxford University Press, 2009.

\bibitem[Rai05]{Raichev}
Alexander Raichev.
\newblock Relative randomness and real closed fields.
\newblock {\em J. Symbolic Logic}, 70(1):319--330, 2005.

\bibitem[Sol75]{Solovay75}
Robert~M. Solovay.
\newblock Handwritten manuscript related to {C}haitin's work.
\newblock IBM Thomas J. Watson Research Center, Yorktown Heights, NY, 215
  pages, 1975.

\end{thebibliography}
\bibliographystyle{alpha}

\end{document}